\documentclass[11pt]{article}

\usepackage{soul}
\usepackage{url}
\usepackage[hidelinks]{hyperref}
\usepackage[utf8]{inputenc}
\usepackage{graphicx}
\usepackage{amsmath}
\usepackage{amsthm}
\usepackage{booktabs}
\usepackage{algorithm2e}
\usepackage[switch]{lineno}
\usepackage[margin=1in]{geometry}
\usepackage{comment}

\usepackage{amsmath,amssymb}
\usepackage{bm}
\usepackage{xcolor}
\usepackage{units}
\usepackage{cite}

\newtheorem{example}{Example}
\newtheorem{theorem}{Theorem}
\newtheorem{lemma}{Lemma}
\newtheorem{proposition}{Proposition}
\newtheorem{corollary}{Corollary}

\usepackage{subcaption}
\usepackage{graphicx}
\usepackage{tikz,units}
\usetikzlibrary{arrows,automata}
\usetikzlibrary{arrows.meta}
\usetikzlibrary{decorations.pathreplacing,decorations.pathmorphing}
\usetikzlibrary{calc}
\usetikzlibrary{patterns,shapes}
\tikzstyle{axis}=[thin]
\tikzstyle{state}=[circle,draw=black,thick,fill=black,minimum size=2mm,inner sep=0mm]

\definecolor{mylightgray}{RGB}{207,209,210}

\newcommand{\RR}{\mathbb{R}}
\newcommand{\Ex}{\mathbb{E}}

\newcommand{\matA}{\mathbf{A}}
\newcommand{\matB}{\mathbf{B}}
\newcommand{\matC}{\mathbf{C}}
\newcommand{\matD}{\mathbf{D}}
\newcommand{\matI}{\mathbf{I}}

\newcommand{\matS}{\mathbf{S}}
\newcommand{\matW}{\mathbf{W}}
\newcommand{\matX}{\mathbf{X}}
\newcommand{\matZ}{\mathbf{Z}}
\newcommand{\matLambda}{\bm{\Lambda}}

\newcommand{\vecz}{\mathbf{z}}
\newcommand{\vecx}{\mathbf{x}}
\newcommand{\vecs}{\mathbf{s}}
\newcommand{\vecq}{\mathbf{q}}

\newcommand{\vecpi}{\boldsymbol{\pi}}

\newcommand{\vecPhi}{\bm{\varphi}}

\newcommand{\calS}{\mathcal{S}}
\newcommand{\calP}{\mathcal{P}}

\newcommand{\classNP}{\textsf{NP}}

\newcommand{\D}{\displaystyle}
\allowdisplaybreaks

\title{Persuading Agents in Opinion Formation Games}

\author{Martin Hoefer\thanks{Dept. of Computer Science, RWTH Aachen University. \texttt{mhoefer@cs.rwth-aachen.de}} \and Tim Koglin\thanks{Dept. of Computer Science, RWTH Aachen University. \texttt{koglin@algo.rwth-aachen.de}} \and Tolga Tel\thanks{Institute for Computer Science, Goethe University Frankfurt. \texttt{tel@em.uni-frankfurt.de}}}

\date{\today}

\begin{document}

\maketitle

\begin{abstract}
        Prominent opinion formation models such as the one by Friedkin and Johnsen (FJ) concentrate on the effects of peer pressure on public opinions. In practice, opinion formation is also based on information about the state of the world and persuasion efforts. In this paper, we analyze an approach of Bayesian persuasion in the FJ model. There is an unknown state of the world that influences the preconceptions of $n$ agents. A sender $\calS$ can (partially) reveal information about the state to all agents. The agents update their preconceptions, and an equilibrium of public opinions emerges. We propose algorithms for the sender to reveal information in order to optimize various aspects of the emerging equilibrium.

    For many natural sender objectives, we show that there are simple optimal strategies. We then focus on a general class of range-based objectives with desired opinion ranges for each agent. We provide efficient algorithms in several cases, e.g., when the matrix of preconceptions in all states has constant rank, or when there is only a polynomial number of range combinations that lead to positive value for $\calS$. This generalizes, e.g., instances with a constant number of states and/or agents, or instances with a logarithmic number of ranges. In general, we show that subadditive range-based objectives allow a simple $n$-approximation, and even for additive ones, obtaining an $n^{1-\varepsilon}$-approximation is \classNP-hard, for any constant $\varepsilon > 0$. 
\end{abstract}
\pagenumbering{gobble}

\clearpage 

\pagenumbering{arabic}

\section{Introduction}
Opinion formation has substantial impact on many aspects of society. With the advent of the internet, computational aspects of opinion formation have become a focus in computer science, especially at the intersection of algorithms and game theory. Prominent opinion formation models from mathematical sociology such as the Friedkin-Johnsen (FJ) model \cite{friedkin1990social} capture notions of peer pressure and the discrepancy between (internal) preconceptions and publicly stated opinions. 

Over the last three decades, an enormous body of work has addressed numerous aspects and extensions of the FJ model and related variants. Interestingly, to the best of our knowledge, natural and straightforward aspects of \emph{persuasion} have not been studied in terms of algorithm design and computational complexity (with the notable exception of Abebe et al.~\cite{abebe2021opinion}), even though these aspects play a crucial role in the formation of opinions. Applications of persuasion include marketing efforts to influence the opinions of potential customers about the quality, usefulness, or social reputation of a certain product. 

For instance, consider a vendor who wants to sell their latest product. The product quality is known to be either high or low, e.g., due to yet volatile manufacturing conditions, with certain probabilities. The potential customers form their individual opinion on the product quality based on their own opinion and the opinions of their peers. Each of them buys the product only if they believe the quality to be high enough. The vendor can make (arbitrary) public announcements to (partially) reveal information about the quality to all potential customers. How should the vendor disclose information in order to sell the product to as many potential customers as possible? 

Similar challenges arise in political campaigns that aim to change the public opinions on a candidate who runs for office. In these applications, there is usually a party (e.g., a marketing agency or interest group) that tries to persuade the public by publishing information in order to influence the opinions of agents in the society on the candidate.

In this paper, we formally model and analyze such persuasion scenarios. More concretely, we consider algorithms for persuasion problems in the FJ model. Our approach lies within the prominent framework of \emph{Bayesian persuasion}~\cite{kamenica2011bayesian}, also termed \emph{information design} or \emph{signaling}. There is a set $\Theta$ of possible states of the world (e.g., whether a product has high or low quality) along with a prior distribution over states. The state of the world influences the preconception of each agent. There is an informed \emph{sender} (e.g., a vendor or marketing company). The sender can see the realized state, while the agents only know the prior. The sender can then (partially) reveal information about the realized state by sending a public signal to all agents. We assume the sender has \emph{commitment power}, i.e., she commits in advance how she sends signals and communicates this \emph{signaling scheme} to the agents. Upon receiving a signal, the agents then Bayes-update their preconceptions according to the conditional distribution over states. Public opinions emerge based on the dynamics inherent in the FJ model. We consider the optimization problems faced by the sender, i.e., designing a signaling scheme that maximizes different aspects of the emerging equilibrium.

On a technical level, our model is closely related to approaches for \emph{persuading voters} that have been of interest over the last decade~\cite{alonso2016persuading,cheng2015mixture,castiglioni2020persuading}. Here, the agents are voters and need to elect a candidate. The utility function of each voter depends on the state of the world. Upon receiving a signal from the sender, each voter chooses the candidate that maximizes the agents' conditional expected utility. The goal of the sender is to, e.g., maximize the probability that a given candidate wins the election, or maximize the expected number of votes for a given candidate. Our model strictly generalizes the elementary voting model with two candidates. Notably, persuasion in the FJ model introduces additional challenges of \emph{externalities} between the agents which are often assumed to be absent in existing work on voter persuasion.

\paragraph{Contribution and Outline.}
In Section~\ref{sec:model}, we describe the precise details of the mathematical model, discuss the relation to voter persuasion and prove some preliminary technical properties. In Section~\ref{sec:simple}, we analyze optimal persuasion for many natural objectives for the sender, including minimizing or maximizing polarization or disagreement in the equilibrium, or norm distance of the equilibrium from a vector of desired opinions. Indeed, all these objectives can be optimized by simple and natural schemes of either revealing full information about the state of the world or no information at all. 
Hence, our results indicate that there are many ``easy'' objectives for persuasion in practice, which allow optimal signaling with very minor conceptual and computational overhead. On a technical level, we observe that many natural objective functions are convex. This allows a connection to pivotal results in the persuasion literature~\cite{kamenica2011bayesian} for the optimality of simple signaling schemes. 

In Section~\ref{sec:range}, we study a computationally more challenging class of functions that we term \emph{range-based objectives}. For such an objective, we are a given a set of intervals (or \emph{ranges}) for each agent. The objective then yields a value based on 
which equilibrium opinion is located in which subset of ranges. Arguably, the easiest such objective is just counting the \emph{number} of ranges in which the equilibrium opinions are located. These objectives extend the standard ones studied for voter persuasion (maximize the number of votes for a candidate, maximize the probability that a candidate is elected). For range-based objectives, we consider two novel algorithms based on cell decomposition approaches. They apply when (1) the states of the world give rise to constantly many linearly independent vectors of preconceptions (Section~\ref{sec:rank}), and (2) the range-based objective has only polynomially many non-zero values (Section~\ref{sec:nonzero}). These scenarios include many natural cases, such as a constant number of agents and/or states, preconceptions that are linearly dependent, or maximizing the probability that all (or a given subset of) agents are located in a particular opinion range. Our results also identify new tractable cases for voter persuasion with two candidates.

In general, we show in Section~\ref{sec:general} that finding a good signaling scheme is \classNP-hard, even maximizing the number of opinions in the ranges. It is \classNP-hard to approximate within a factor of $n^{1-\varepsilon}$, where $n$ is the number of agents and  $\varepsilon > 0$ is an arbitrary constant. We show a simple and matching $n$-approximation for subadditive range-based objectives by applying our results from Section~\ref{sec:nonzero}. We optimize the signaling scheme in terms of the objective for each individual agent and then use one of these $n$ schemes that yields the best objective function value overall. 

In Section~\ref{sec:extensions}, we elaborate that our results can be applied to analyze signaling in other well-studied variants of the FJ model, such as multi-dimensional opinions \cite{parsegov2016novel}. Moreover, we show improved results for signaling in the classic French-DeGroot model~\cite{degroot1974reaching,French1956}, where preconceptions are absent and agents reach a consensus. In particular, our results from Section~\ref{sec:rank} can be adapted to compute optimal signaling schemes for range-based objectives in polynomial time.

\paragraph{Related Work.}

Our paper studies Bayesian persuasion in the FJ model. Both the literature on persuasion and on the FJ model are too vast to survey here. We concentrate on a selected subset of works that are most closely related.

The study of Bayesian persuasion was popularized by Kamenica and Gentzkow \cite{kamenica2011bayesian} and later adopted prominently by Bergemann and Morris~\cite{bergemann2016information,bergemann2019information}. Following Dughmi and Xu~\cite{DughmiX21}, Bayesian persuasion has attracted significant interest in computer science. 

Closest to our approach is previous work on persuading voters initiated by Alonso and C\^amera~\cite{alonso2016persuading}. Here, an individual (politician) can influence voters' choices by designing a public signal. The authors characterized the optimal public signal where the individual designs a public signal with realizations targeting different winning coalitions. Designing optimal public signals in general multi-receiver settings is known to be hard in terms of approximation~\cite{Rubinstein17,Dughmi19}. This has also been shown for an elementary voting scenario by Cheng et al.~\cite{cheng2015mixture} in the context of a more general approach of mixture selection. For a suitably normalized instance, an approximation algorithm with additive guarantees is given when the instances satisfies conditions on Lipschitz continuity and noise stability.

Several extensions of the elementary voting model have found interest. Chan et al.\cite{chan2019pivotal} show that in a voting context a sender can exploit the heterogeneity in voting costs by privately communicating with the voters. Multiple candidates and district-based voting are studied in \cite{castiglioni2020persuading,castiglioni2021persuading} with private, public and semi-public signaling schemes. These results are complementary to \cite{arieli2019private,xu2020tractability,CastiglioniCG23}, who provide fixed-parameter tractable and (bi-criteria) approximation algorithms for private and public signals in multi-receiver persuasion without externalities. More generally, our work extends the algorithmic study of persuasion with multiple agents and externalities, which has received interest in a number of applications including, e.g., auctions~\cite{BacchiocchiCMRG22,CastiglioniMRG23,EmekFGLT14,MiltersenS12} or traffic routing~\cite{ZhouNX22,CastiglioniCMG21,BhaskarCKS16,GriesbachHKK22,GriesbachHKK24}

In terms of opinion formation models, we build on the classic Friedkin-Johnsen model for continuous opinions~\cite{friedkin1990social}. Convergence to an equilibrium and its dependence on the network structure is analyzed by Ghaderi and Srikant~\cite{ghaderi2014opinion}. These results extend to multi-dimensional opinions \cite{parsegov2016novel}. For a survey of structural features of opinion formation models (including the FJ model), we refer to Proskurnikov and Tempo \cite{ProskunikovT2017}.

Most closely related to our work is work on persuasion in the FJ model by Abebe et al.~\cite{abebe2021opinion}, where a designer can modify the susceptibility of agents. The authors propose a framework in which the self-confidence parameter of the agents should be optimized, thereby changing the extent each agent suffers from peer-pressure. This approach is quite different from our model, in which there are different states of the world, and a sender can exploit stochastic uncertainty to influence the resulting opinions of the agents.

As part of the objective of the sender, we are interested in structural aspects of the emerging opinions including polarization. Algorithms to manipulate the network structure to reduce mutual disagreement and polarization in the resulting equilibrium are analyzed by Musco et al.~\cite{musco2018minimizing}. More generally, computational and structural aspects of opinion formation models have been a very popular topic in computer science over the last decade. There is a large body of work on, e.g., discrete opinion spaces~\cite{ChierichettiKO18,AulettaCFGP16,MeiHCBD24}, consensus properties~\cite{AulettaFG18,FerraioliV17}, manipulation of individual preconceptions~\cite{Sun2023optimization}, limited information exchange between agents~\cite{FotakisKKS18}, combinations with voting~\cite{BerenbrinkHKMRW25,EpitropouFHS19,BredereckGIK22} and more.

\section{Model and Preliminaries}
\label{sec:model}

We consider the classic FJ opinion formation model as formulated in~\cite[Chapter 6]{ProskunikovT2017}. There is a set $V$ of $n$ agents embedded in a social influence network. Network influence is expressed by a non-negative, stochastic $n \times n$-matrix $\matA$, where entry $a_{uv} \ge 0$ yields the influence of $v$ on $u$.

Each agent $u \in V$ has a \emph{preconception} $s_u \in [0,1]$, which is static. The \emph{stubbornness}, captured by $(1-\lambda_u)$ for $\lambda_u \in [0,1]$, indicates to which extent $u$ prefers to stick to their preconception. Thus, $\lambda_u$ states the \emph{susceptibility} of $u$ to social influence in the process.

Agent $u$ expresses a \emph{public opinion} that is governed by a mix of stubbornness and peer pressure. The formation process of public opinions starts at $z_u^{(0)} = s_u$ for each $u \in V$, and then evolves via simultaneous updates in discrete rounds $t=1,2,3\ldots$ according to
\begin{equation}
    \label{eq:FJ_update}
    z_u^{(t+1)} = (1-\lambda_u) s_u +  \lambda_u \sum_{v \in V} a_{uv} z_v^{(t)} \, .
\end{equation}
Using $\vecz^{(t)} = (z_1^{(t)},\ldots,z_n^{(t)})^\top$, $\vecs = (s_1,\ldots,s_n)^\top$ and $\matLambda = diag(\lambda_1,\ldots,\lambda_n)$, the update step is given by $\vecz^{(t+1)} = (\matI - \matLambda) \vecs + \matLambda \matA \vecz^{(t)} = \matD\vecs  + \matW \vecz^{(t)}$, where we use $\matD = \matI - \matLambda$ and $\matW = \matLambda\matA$.

Provided that $\matW$ and $\matD$ satisfy certain structural conditions, the process is guaranteed to converge to a unique equilibrium. For a discussion of these conditions, we refer the reader to Proskurnikov and Tempo~\cite{ProskunikovT2017}. The equilibrium is then given by $\vecz = \lim_{t \to \infty} \vecz^{(t)} = \sum_{t=0}^\infty \matW^{t} \matD \vecs = (\matI-\matW)^{-1} \matD \vecs$.
To enable a meaningful prediction of the emerging behavior of the agent population, we focus on instances that allow for convergence.

We consider \emph{Bayesian persuasion} in the FJ model, which we term an \emph{FJ signaling game}. There is a set $\Theta$ of $m$ different states of nature. State $\theta \in \Theta$ determines the preconceptions of the agents. Each agent $u \in V$ has a preconception $s_{u,\theta} \in [0,1]$, for each $\theta \in \Theta$. We denote the $(n \times m)$-matrix of preconceptions by $\matS$. There is a prior distribution over the states. We use $\vecq = (q_1,\ldots,q_m)^\top$ to denote the vector of probabilities. The set of states $\Theta$, the distribution $\vecq$, and her vector of potential preconceptions $\vecs_u$ are known to agent $u \in V$ upfront. However, $u$ is unable to see the realization of the random state.

There is a sender $\calS$ (say, a political campaign, or a marketing company) who observes the realized state $\theta$. $\calS$ can (partially) disclose information about $\theta$ to the agents. In particular, there is a set $\Sigma$ of abstract \emph{signals}. After observing the realized state, $\calS$ sends one (public) signal $\sigma$ to all agents. Seeing $\sigma$, the agents perform a Bayes update of their preconceptions and then initiate the FJ opinion formation process. 

Formally, we assume $\calS$ has commitment power. Thus, $\calS$ initially commits to a \emph{signaling scheme} $\vecPhi$. It is a matrix that contains a probability $\varphi_{\theta,\sigma} \in [0,1]$ to send signal $\sigma$ when the realized state is $\theta$, for each $\theta \in \Theta$, $\sigma \in \Sigma$. $\calS$ communicates $\vecPhi$ to all agents. Then the state $\theta$ is realized and observed by $\calS$. Due to commitment, $\calS$ is bound to send signals according to distribution $\vecPhi_{\theta}$. Seeing signal $\sigma$ (but not knowing $\theta$), each agent $u$ performs a Bayes update on the distribution over states. This yields an expected preconception
\[
    s_{u,\sigma} = \frac{\sum_{\theta \in \Theta} s_{u,\theta} x_{\theta,\sigma} }{\sum_{\theta \in \Theta} x_{\theta,\sigma}} \, ,
\]
where we use $x_{\theta,\sigma} = \varphi_{\theta,\sigma} \cdot q_\theta$ to denote the combined probability that state $\theta$ is realized and signal $\sigma$ gets sent. Let $\matX = (x_{\theta,\sigma})_{\theta \in \Theta, \sigma \in \Sigma}$ denote the matrix and $\vecx_{\sigma} = (x_{\theta,\sigma})_{\theta \in \Theta}$ the column vector for signal $\sigma$. The Bayes update and the emerging expected preconceptions can be expressed by 
\begin{equation}
    \label{eq:privateBayes}
    \vecs_{\sigma} = \frac{\matS\vecx_{\sigma}}{\| \vecx_{\sigma} \|_1} \, ,
\end{equation}
where $\|\vecx_\sigma\|_1$ represents the probability that signal $\sigma$ gets sent, and $\vecx_{\sigma}/\| \vecx_{\sigma} \|_1$ the posterior distribution over states induced by $\sigma$. Given the expected preconceptions $\vecs_{\sigma}$, the FJ formation process then converges to the equilibrium $\vecz_{\sigma} = (\matI - \matW)^{-1} \matD \vecs_{\sigma}$. For convenience, we define matrix $\matZ = (\matI-\matW)^{-1}\matD \matS$ of all equilibrium opinions $z_{u,\theta}, u \in V, \theta \in \Theta$, that evolve when $\calS$ uses ``full revelation'' signaling, i.e., sends in each state $\theta_i$ an individual signal $\sigma_i$ with probability 1. We term $\matZ$ the \emph{full revelation matrix}.

\begin{example} \rm We discuss the introductory example in more detail to illustrate the model and definitions. Consider a vendor $\calS$ who tries to convince two potential customers to buy her product. We cast the scenario as an FJ signaling game with $V = \{u,v\}$ and two directed edges with $a_{uv} = a_{vu} = 1$. Let $\lambda_u = \lambda_v = 0.5$. 
Suppose there are two states $\Theta = \{\theta_1 , \theta_2\}$ with $\vecq = (0.5,0.5)^\top$. In state $\theta_1$ the product has low quality, in $\theta_2$ it has high quality. The preconceptions are $\vecs_{\theta_1} = (0,0.3)^\top$ and $\vecs_{\theta_2}=(1,0.7)^\top$. The numbers express the level of approval, i.e., both agents prefer high quality, and agent 2 has a more moderate opinion of low quality. Now $\calS$ decides on a signaling scheme. After seeing the state, $\calS$ sends a signal to both agents, who Bayes-update their preconceptions (i.e., their approval rating of the product). Then the agents interact and reach an equilibrium $\vecz$.

Suppose $\calS$ uses a ``no-signaling'' scheme, in which she always sends the same signal $\sigma$. Then $\vecx_\sigma = \vecq$ and $\|\vecx_\sigma\|_1 = 1$, i.e., the posterior distribution is the prior. The expected preconceptions become $\vecs_\sigma = (0.5,0.5)^\top$, which is also the resulting equilibrium $\vecz_\sigma = \vecs_\sigma$. 

Instead, suppose $\calS$ uses full revelation. Then the posterior distributions are $\vecx_{\sigma_1}/\|\vecx_{\sigma_1}\|_1 = (1,0)^\top$ and $\vecx_{\sigma_2}/\|\vecx_{\sigma_2}\|_1 = (0,1)^\top$. The preconceptions are $\vecs_{\sigma_1} = \vecs_{\theta_1} = (0,0.3)^\top$ and $\vecs_{\sigma_2} = \vecs_{\theta_2} = (1,0.7)^\top$. The emerging equilibria are given by the full revelation matrix $\matZ$ with columns $\vecz_{\theta_1} = (0.1,0.2)^\top$ and $\vecz_{\theta_2} = (0.9,0.8)^\top$. Each state/signal occurs with probability $\|\vecx_{\sigma_1}\|_1 = q_{\theta_1} = \| \vecx_{\sigma_2}\|_1 = q_{\theta_2} = 0.5$. \hfill $\blacktriangle$
\end{example}

We consider the problem to design $\vecPhi$ (and, hence, $\matX$) in order to induce a distribution over equilibrium outcomes $\vecz_{\sigma}$. We consider optimizing several natural objective functions to evaluate the equilibria $\vecz_\sigma$. Given an objective $f(\vecz_\sigma)$, we strive to find $\vecPhi$ that optimizes $\Ex[f(\vecz_\sigma)] = \sum_{\sigma} \|\vecx_\sigma\|_1 f(\vecz_\sigma)$. We consider the following classes of objectives.

\paragraph{Expected Norm Distance.}
$\calS$ has a vector $\vecz^*$ with a desired (undesired) target opinion $z_u^*$ for each agent $u \in V$. $\calS$ strives to minimize (maximize) the expected $p$-norm distance, for some $p \in \{1,2,\ldots,\infty\}$, from this target vector, i.e., $f_{nd}(\vecz_{\sigma}) = \|\vecz_{\sigma} - \vecz^* \|_p$.

\paragraph{Polarization and Disagreement Indices.} $\calS$ strives to minimize (maximize) the \emph{polarization index} $f_{p}(\vecz_\sigma) = \sum_{u \in V} (z_{u,\sigma} - \overline{\vecz_{\sigma}})^2$, where $\overline{\vecz_\sigma} = \frac{1}{n}\sum_{u \in V} z_{u,\sigma}$ is the average opinion, or the \emph{disagreement index} $f_d(\vecz_\sigma) = \sum_{(u,v) \in E} w_{uv}(z_{u,\sigma} - z_{v,\sigma})^2$ ~\cite{musco2018minimizing}. We also consider the natural variants \emph{max-polarization} and \emph{max-disagreement}, in which $\calS$ minimizes (maximizes) $f_{mp}(\vecz_\sigma) = \max_{u,v \in V} | z_{u,\sigma} - z_{v,\sigma} |$ or $f_{md}(\vecz_\sigma) = \max_{(u,v) \in E} | z_{u,\sigma} - z_{v,\sigma} |$.

\paragraph{Desired Ranges.}
$\calS$ has desired opinion ranges for each agent. There is a finite set $D_u$ of intervals (or \emph{ranges}) for each $u \in V$. $\calS$ strives to maximize (minimize) the expected number of agents with an opinion in one of its ranges, i.e., $f_{na}(\vecz_{\sigma}) = |\{ u \mid \text{there is } [a,b] \in D_u \text{ with } z_{u,\sigma} \in [a,b]\}|$.

Beyond counting the \emph{number of agents} whose opinion is in one of their desired intervals, we study a substantially more general class of \emph{range-based objectives}. A range-based objective is based on a \emph{monotone\footnote{$g$ is monotone if $g(T') \le g(T)$ for every $T' \subseteq T \subseteq D$.} set function} $g : 2^D \to \RR_{\ge 0}$, where $D = \{ i_{u,[a,b]} \mid u \in V, [a,b] \in D_u\}$ is a set of indicator elements if agent $u$ is located in opinion range $[a,b] \in D_u$. $\calS$ strives to maximize $f(\vecz_\sigma) = g(T_\sigma)$, where $T_\sigma = \{ i_{u,[a,b]} \mid u \in V, z_{u,\sigma} \in [a,b] \in D_u \}$. We denote by $k = |D|$ the total number of ranges for all agents.

\setcounter{example}{0}
\begin{example}{(continued)} \rm
    Reconsider the example above. Suppose the target opinions for $\calS$ are $\vecz^* = (0.7,0.7)^\top$, and $\calS$ strives to minimize a 2-norm-distance. For the no-signaling scheme, we obtain a cost of $1 \cdot \| (0.5,0.5)^\top - (0.7,0.7)^\top\|_2 = \sqrt{0.08}\approx 0.282$. For full revelation, the cost is $0.5\cdot\| (0.1,0.2)^\top - (0.7,0.7)^\top\|_2 + 0.5\cdot\| (0.9,0.8)^\top - (0.7,0.7)^\top\|_2 = 0.5\cdot\sqrt{0.61}+0.5\cdot\sqrt{0.05}\approx 0.502$. Hence, no-signaling is the better scheme. We show below that it is even an \emph{optimal} scheme for $\calS$.

    Alternatively, suppose each agent buys the product if his approval exceeds $0.6$, i.e., the opinion range desired by $\calS$ is $[0.6,1]$. $\calS$ wants to maximize the probability that both agents buy the product. Generally, for any signal $\sigma$ we have
    \[
    \vecz_\sigma = (\matI -\matW)^{-1} \matD \frac{\matS \vecx_\sigma}{\|\vecx_\sigma\|_1} = \left(\begin{array}{cc} 0.1 & 0.9 \\ 0.2 & 0.8\end{array}\right) \left(\begin{array}{c} \nicefrac{x_{\theta_1,\sigma}}{\| \vecx_\sigma\|_1}\\ \nicefrac{x_{\theta_2,\sigma}}{\| \vecx_\sigma \|_1} \end{array}\right) \, .
    \]
    An optimal scheme is given by, e.g., $x_{\theta_1,\sigma_1} = 1/4$ and $x_{\theta_2,\sigma_1} = 1/2$ and $x_{\theta_1,\sigma_2} = 1/4$ and $x_{\theta_2,\sigma_2} = 0$. This yields $\vecz_{\sigma_1} = (0.6\overline{3},0.6)^\top$ and $\vecz_{\sigma_2} = (0.1,0.2)^\top$, as well as $\|\vecx_{\sigma_1}\| = 0.75, \| \vecx_{\sigma_2}\| = 0.25$. For signal $\sigma_1$, both agents are in the desired interval. For signal $\sigma_2$, no agent is in the desired interval. Thus, in expectation, 1.5 agents are in the desired interval. Hence, the vendor can sell in expectation $75\%$ of her products, even though people were apriori interested in only $50\%$ of the products (with high quality). \hfill $\blacksquare$
    
\end{example}

\paragraph{Persuading Voters.} FJ signaling games with range-based objectives is a broad generalization of the elementary scenario of persuading voters~\cite{alonso2016persuading,cheng2015mixture}, where the agents are voters. There is an issue to be decided upon via majority approval. Each state determines a utility pair for each agent, with a utility for a yes- or no-decision. The sender sends a signal, voters Bayes-update their preconceptions, and then vote based on which decision would maximize their expected utility. The final decision is given by the majority of votes. A goal of the sender is, e.g., to maximize the probability of a yes-decision in the end. 

Via standard normalization, one can assume that in each state, each voter has utility $0.5$ for no and some utility in $[0,1]$ for yes. Thus, $\calS$ strives to maximize the probability that at least $n/2$ agents have a utility of at least $0.5$ for yes. This scenario can be expressed as an FJ signaling game using $E = \emptyset$ and $s_{u,\theta} \in [0,1]$ as the utility for yes. We set $D_u = \{[0.5,1]\}$ for all $u \in V$, as well as $g(T) = 1$ if $|T| \ge n/2$ and $0$ otherwise. We maximize $g$ by maximizing the probability that at least $n/2$ agents vote yes.

\paragraph{Preliminaries.}
To conclude the section, we discuss preliminaries that will be useful below.

\begin{lemma} \label{lem:equilib}
    Consider any FJ signaling game and its full revelation matrix $\matZ$. There is an equivalent game with the same states $\Theta$ and prior distribution $\vecq$, in which preconceptions $s_{u,\theta} = z_{u,\theta}$, for all $u \in V$ and $\theta \in \Theta$, as well as $\matD = \matI$ and $\matW = \bm{0}$.
\end{lemma}

\begin{proof}
    Consider any signaling scheme $\vecPhi$. For each signal $\sigma \in \Sigma$ we obtain the vector of preconceptions $\vecs_{\sigma}$ from \eqref{eq:privateBayes}. The resulting equilibrium opinions are
    \begin{align}
    \label{eq:opinionDecomp}
        \vecz_\sigma &= 
        (\matI - \matW)^{-1} \matD \, \frac{\matS \vecx_{\sigma}}{\| \vecx_{\sigma}\|_1} = \frac{\matZ \vecx_{\sigma}}{\| \vecx_{\sigma}\|_1} \, .
    \end{align}
 The equilibrium opinions $\vecz_\sigma$ can be interpreted as a linear combination of the equilibrium opinions in $\matZ$. Considering the equivalent game with preconceptions $\matS = \matZ$, $\matW = \bm{0}$ and $\matD = \matI$, we see that $\matZ$ trivially remains the full revelation matrix. Moreover, by~\eqref{eq:opinionDecomp} the equilibrium opinions $\vecz_\sigma$ remain the same, for any signaling scheme $\varphi$. 
\end{proof}

\begin{lemma}
    \label{lem:rank}
    The matrices $\matS$ and $\matZ$ have the same rank.
\end{lemma}
\begin{proof}
    Both $\matS$ and $\matZ$ are $(n\times m)$-matrices, and their rank is at most $\min(n,m)$. Both $(\matI-\matW)^{-1}$ and $\matD$ are invertible square matrices. As such, $(\matI-\matW)^{-1}\matD$ is also an invertible square matrix and, thus, has full rank $n$. Now since $\matZ = ((\matI-\matW)^{-1} \matD) \matS$, the rank of $\matZ$ is equal to the rank of $\matS$. 
\end{proof}

\section{Optimal Simple Schemes}
\label{sec:simple}

In this section, we show that many natural objectives for $\calS$ are optimized using two simple schemes, \emph{no-signaling} or \emph{full revelation}. For no-signaling, we assume that $\calS$ always sends the same signal, i.e., there is $\sigma$ such that $\varphi_{\theta,\sigma} = 1$ for all $\theta \in \Theta$. For full revelation, $\calS$ sends an individual signal for each state, i.e., for each $\theta \in \Theta$ there is a signal $\sigma \in \Sigma$ such that $\varphi_{\theta,\sigma} = 1$, and $\varphi_{\theta',\sigma} = 0$ for all $\theta' \in \Theta$ where $\theta' \neq \theta$.

A key property is that many objective functions $f(\vecz_\sigma)$ are \emph{convex}, i.e., they satisfy $f(\delta \vecz + (1-\delta)\vecz') \le \delta f(\vecz) + (1-\delta) f(\vecz')$, for any $\vecz, \vecz'$ and $\delta \in [0,1]$. For these functions, simple schemes are optimal.

\begin{theorem}
    \label{thm:convex}
    If $\calS$ wants to minimize a convex objective function $f(\vecz_\sigma)$, then no-signaling is an optimal scheme. If $\calS$ wants to maximize such a function, then full revelation is an optimal scheme.
\end{theorem}
The theorem follows directly from classic results on Bayesian persuasion by Kamenica and Gentzkow~\cite{kamenica2011bayesian}. We present a proof for completeness.

\begin{proof}[Proof of Theorem~\ref{thm:convex}]
     Consider matrix $\matZ$, and suppose $\vecz_{\theta}$ is the column vector of equilibrium opinions for state $\theta$. Consider an arbitrary signaling scheme $\vecPhi$. For every signal $\sigma$, we observe that $\vecz_\sigma = \matZ\vecx_\sigma/\|\vecx_\sigma\|_1 = \sum_{\theta \in \Theta} \frac{x_{\theta,\sigma}}{\|\vecx_\sigma\|} \vecz_{\theta}$. Using this, we bound the expected value of $\vecPhi$ by
    \begin{align*}
        &\Ex[f(\vecz_\sigma)] = \sum_{\theta,\sigma} x_{\theta,\sigma} \cdot f(\vecz_\sigma)
        = \sum_{\theta,\sigma} x_{\theta,\sigma} \cdot f\left(\sum_{\theta} \frac{x_{\theta,\sigma}}{\|\vecx_\sigma\|_1} \vecz_{\theta}\right) \\
        &\le \sum_{\theta,\sigma} x_{\theta,\sigma} \cdot \left(\sum_{\theta} \frac{x_{\theta,\sigma}}{\|\vecx_\sigma\|_1} f(\vecz_{\theta})\right) = \sum_{\theta,\sigma} x_{\theta,\sigma} f(\vecz_{\theta})\\
        &=  \sum_{\theta \in \Theta} q_{\theta} \sum_{\sigma} \varphi_{\theta,\sigma} f(\vecz_{\theta \in \Theta}) = \sum_{\theta} q_\theta f(\vecz_{\theta})\enspace,
    \end{align*}
    where the inequality follows from convexity of $f$. The final expression is the expected cost of full revelation: With probability $q_\theta$ state $\theta$ arises, $\calS$ fully reveals $\theta$, and $\vecz_\theta$ evolves.

    Similarly,
    \begin{align*}
        &\Ex[f(\vecz_\sigma)] = \sum_{\theta,\sigma} x_{\theta,\sigma} \cdot f\left(\sum_{\theta} \frac{x_{\theta,\sigma}}{\|\vecx_\sigma\|_1} \vecz_{\theta}\right) \\
        &\ge f\left(\sum_{\theta,\sigma} x_{\theta,\sigma} \cdot \left(\sum_{\theta} \frac{x_{\theta,\sigma}}{\|\vecx_\sigma\|_1} \vecz_{\theta} \right)\right) = f\left(\sum_{\theta \in \Theta} q_\theta \vecz_{\theta}\right)\\
        &= f(\matZ \vecq) = f((\matI-\matW)^{-1} \matD \matS \vecq) \enspace.
    \end{align*}
    The last expression is the cost of no-signal: $\matS \vecq$ is the a-priori expectation of preconceptions. This is used for the FJ formation process, and the equilibrium $\matZ \vecq$ emerges. 
\end{proof}

We observe that a variety of distance-based objectives for $\calS$ fulfill convexity. Thus, when $\calP$ strives to optimize any (convex combination) of these objectives, Theorem~\ref{thm:convex} applies and proves optimality of simple signaling schemes. More precisely, no-signaling is optimal when minimizing these functions, and full revelation is optimal when maximizing them.
\begin{lemma}
    $f_{nd}$, $f_p$, $f_d$, $f_{mp}$, and $f_{md}$ are all convex.
\end{lemma}
\begin{proof}
    By definition, every norm is convex, and every $p$-norm is indeed a norm. $\vecz_{\sigma} - \vecz^*$ is linear in $\vecz_\sigma$, so $f_{nd}(\vecz_{\sigma}) = \|\vecz_{\sigma} - \vecz^* \|_p$ is convex.
    
    The average opinion $\overline{\vecz_{\sigma}}$ in any given signal $\sigma$ is a linear combination of average opinions from different states. Hence, $z_{u,\sigma} - \overline{\vecz_{\sigma}}$ is linear in $z_{u,\sigma}$. Now $f(x) = x^2$ is a convex function, so $(z_{u,\sigma} - \overline{\vecz_{\sigma}})^2$ is convex. The sum of convex functions is also a convex function, hence, $f_{p}(\vecz_\sigma) = \sum_{u \in V} (z_{u,\sigma} - \overline{\vecz_{\sigma}})^2$ is convex. 
    
    For $f_d(\vecz_\sigma) = \sum_{(u,v) \in E} w_{uv} (z_{u,\sigma} - z_{v,\sigma})^2$ the argument is similar. Each $(z_{u,\sigma} - z_{v,\sigma})^2$ is convex in each variable. The sum of convex functions is a convex function.
    
    Both, $f_{mp}(\vecz_\sigma) = \max_{u,v \in V} | z_{u,\sigma} - z_{v,\sigma} |$ and $f_{md}(\vecz_\sigma) = \max_{(u,v) \in E} | z_{u,\sigma} - z_{v,\sigma} |$ return the maximum value of convex functions $f(x,y) = |x-y|$. They are convex.
\end{proof}

On an intuitive level, this result can be interpreted as follows. The peer pressure effect in the FJ model introduces a tendency towards consensus in equilibrium. Consensus opinions minimize many convex objectives, such as polarization or disagreement indices. In contrast, preconceptions introduce heterogeneity and discrepancy in equilibrium. As a signaling scheme gets more informative, agents can Bayes update their intrinsic differences in a more pronounced fashion. Heterogeneity in preconceptions emerges more drastically in the resulting equilibrium. Consequently, no-signaling and full revelation are the schemes using which the impact of preconceptions (and hence, e.g., polarization and disagreement indices) in the emerging equilibrium are minimized and maximized, respectively.

\section{Range-Based Objectives}
\label{sec:range}

\subsection{Rank of \textbf{S}}
\label{sec:rank}

In this section, we provide a polynomial-time algorithm for optimal signaling when the $(n \times m)$-matrix $\matS$ of preconceptions has constant rank $d=O(1)$. This scenario generalizes many natural special cases, such as, e.g.:
\begin{itemize}
    \item Constant number of agents and/or states.
    \item Constant number of agents with state-de\-pen\-dent preconceptions (any other agent has the same preconception in all states).
    \item Constant number of agents with preconceptions (any other agent has no preconception at all). 
    \item There is a basic preconception $s^b_u$ for each agent $u \in V$. In every state $\theta$, this preconception gets scaled by a factor $\alpha_\theta > 0$, i.e., $s_{u,\theta} = s^b_u \cdot \alpha_\theta$.
\end{itemize}

\begin{theorem}
    \label{thm:constantRank}
    If $\matS$ has constant rank, then there is a polynomial-time algorithm to compute an optimal signaling scheme for every range-based objective.
\end{theorem}

\begin{proof}
    Suppose $\matS$ has constant rank $d$, then $\matZ$ also has rank $d$ by Lemma~\ref{lem:rank}. We decompose $\matZ = \matB \matC$ using an $(n \times d)$-matrix $\matB = (b_{u,j})_{u\in V, j \in [d]}$ whose column vectors are a basis of the space spanned by $\matZ$, where we use the notation $[d] = \{1,\ldots,d\}$. The $(d \times m)$-matrix $\matC = (c_{j,\theta})_{j \in [d], \theta \in \Theta}$ contains the linear combination of the basis vectors required to generate $\matZ$. For any signal $\sigma$, we define $\vecx^{(b)}_\sigma = \matC \vecx_\sigma/\|\vecx_\sigma\|_1$ (which might be seen as a ``low-dimensional posterior'') and see that $\vecz_\sigma = \matB \vecx^{(b)}_\sigma$.

    For each agent $u \in V$ and each interval $[a,b] \in D_u$, we can partition the space of possible signaling schemes according to the following (in-)equalities: $z_{u,\sigma} \; \Box \; \alpha$ with $\Box \in \{>, =, < \}$ and $\alpha \in \{a,b\}$. Feasible combinations of these six inequalities divide the space into five regions: (1) $z_{u,\sigma} > b$, (2) $z_{u,\sigma} = b$, (3) $z_{u,\sigma} \in (a,b)$, (4) $z_{u,\sigma} = a$ and (5) $z_{u,\sigma} < a$.

    Recall that we have $k$ intervals in total for all agents. These intervals introduce at most $6k$ inequalities. By introducing $r$ hyperplanes into $\RR^d$, the space gets divided into $O(r^d)$ different cells~\cite{buck1943partition}. A cell here refers to a set of points in $\RR^d$ that satisfy exactly the same (in-)equalities (and violate all others). Each cell is convex. 
    
    Indeed, we can interpret $\vecz_\sigma$ as a vector in $d$-dimensional space, since $\vecz_\sigma = \matB\vecx^{(b)}_\sigma$ emerges from the low-dimensional posterior. We introduce $6k$ inequalities into $\RR^d$, which leads to $c = O(k^d)$ many cells. They can be enumerated in polynomial time using classic reverse search algorithms~\cite{avis1996reverse}. For each cell, we detect which of the $6k$ (in-)equalities are satisfied by this cell. This yields a description of the linear polytope $\calP_i$ of each cell $i \in [c]$.

    The polytope $\calP_i$ determines for each agent the set of ranges its' equilibrium opinion is located in. As such, $f(\vecz_\sigma)$ is the same for each point $\vecz_\sigma \in \calP_i$. We denote this value by $f(\calP_i)$. Suppose there are two signals $\sigma_1, \sigma_2$ with $\vecz_{\sigma_1},\vecz_{\sigma_2} \in \calP_i$. We can adjust $\vecPhi$ such that whenever $\calS$ decides to send $\sigma_1$ and $\sigma_2$, it sends a new signal $\sigma$ instead. This new signal has probability $\| \vecx_{\sigma_1}\|_1 + \|\vecx_{\sigma_2}\|_1$. Since the cell is convex, $\vecz_\sigma \in \calP_i$ and, hence, the value $f(\vecz_\sigma) = f(\calP_i)$. The adjusted signaling scheme using $\sigma$ instead of $\sigma_1,\sigma_2$ yields the same expected value. Consequently, w.l.o.g.\ we assume for each cell $i\in [c]$, $\vecPhi$ has at most one signal $\sigma_i$ with $\vecz_{\sigma_i}$ in that cell.

    Consider any cell $i$ with an equality constraint w.r.t.\ range $[a,b]$, i.e., $\vecz_{u,\sigma} = a$ or $\vecz_{u,\sigma} = b$. We can replace each of these equality constraints by the pair of constraints $\vecz_{u,\sigma} \ge a$ and $\vecz_{u,\sigma} \le b$. This does not change the inclusion status w.r.t.\ range $[a,b]$ and, consequently, the value of $f(\calP_i)$. Consequently, we only consider cells based on feasible combinations of $z_{u,\sigma} \; \Box \; \alpha$ with $(\Box,\alpha) \in \{(<,a), (\ge,a), (\le,b), (>,b)\}$. Further, observe that the polytope $\calP_i$ for a cell $i$ is open whenever the cell is characterized by a strict inequality. We enlarge each cell to include its boundary, i.e., we replace each strict inequality with $>$ or $<$ by a $\le$ or $\ge$, respectively. A cell may then include solutions that do not yield a value\footnote{For example, the region described by the inequality $z_{u,\sigma} \ge b$ can have a different value for $f(\vecz_\sigma)$ when the original $z_{u,\sigma} > b$ holds or when the solution is on the boundary $z_{u,\sigma} = b$.} $f(\calP_i)$. These solutions on the boundary of the polytope are then shared with other cells. Suppose $\vecz_{\sigma}$ is on the boundary of several cells. Since we strive to maximize $f(\vecz_\sigma)$, we interpret this signal to lie in the cell $\sigma_i$ that yields the maximum value $f(\calP_i)$. As $f$ is \emph{monotone non-decreasing}, this is the cell $\calP_i$ that correctly reflects all interval inclusions. 
    
    Given closed polytopes $\calP_i$ and at most one signal $\sigma_i$ for each cell $i \in [c]$, we obtain an optimal signaling scheme by solving the optimization problem~\eqref{eq:LP1}. The first constraints ensure that $\vecPhi$ yields a distribution over signals for each $\theta$. The next constraints define the entries of $\matX$ and $x_{\sigma_i} = \| \vecx_{\sigma_i} \|_1$, the low-dimensional posteriors $\vecx_{\sigma_i}^{(b)}$, and the equilibrium vectors $\vecz_{u,\sigma_i} = \matB \vecx_{\sigma_i}^{(b)}$. Finally, we add the constraints describing each (closed) polytope $\calP_i$. 
        
    \begin{equation}
    \label{eq:LP1}
    \begin{array}{lrcll}
        \multicolumn{4}{l}{\text{Max.\ } \D \sum_{i=1}^c x_{\sigma_i} \cdot f(\calP_i)}\\[0.1cm]
        \text{s.t.\ } & \D \sum_{i=1}^c \varphi_{\theta,\sigma_i} &=& 1 & \text{ for all } \theta \in \Theta\\
                    & \varphi_{\theta,\sigma_i} & \ge & 0 & \text{ for all } \theta \in \Theta, i \in [c] \\
                    & x_{\theta,\sigma_i} &=& \varphi_{\theta,\sigma_i} \cdot q_\theta & \text{ for all } \theta \in \Theta, i \in [c] \\
                    & x_{\sigma_i} &=& \D \sum_{\theta \in \Theta} x_{\theta,\sigma_i} & \text{ for all } i \in [c] \\
                    & x_{j,\sigma_i}^{(b)} &=& \D \sum_{\theta \in \Theta} c_{j,\theta} \cdot \D \frac{x_{\theta,\sigma_i}}{x_{\sigma_i}} & \text{ for all } j \in [d]\\
                    & z_{u,\sigma_i} &=& \D \sum_{j=1}^d b_{\theta,j} x_{j,\sigma_i}^{(b)}  & \text{ for all } u \in V, i \in [c]\\
                    & z_{u,\sigma_i} &\in& \calP_i & \text{ for all } u \in V, i \in [c]
    \end{array}    
    \end{equation}
    
    The program~\eqref{eq:LP1} is not an LP, since $\vecx_{\sigma_i}^{(b)}$ is given by a non-linear constraint. An equivalent formulation with $y_{j,\sigma_i} = x_{j,\sigma_i}^{(b)} x_{\sigma_i}$ that replaces the last three constraint classes is 
    
    \begin{equation*}
    \begin{array}{lrcll}
                    & y_{j,\sigma_i} &=& \D \sum_{\theta \in \Theta} c_{j,\theta} \cdot x_{\theta,\sigma_i} & \text{ for all } j \in [d]\\
                    & z_{u,\sigma_i} &=& \D \sum_{j=1}^d b_{\theta,j} y_{j,\sigma_i} & \text{ for all } u \in V, i \in [c]\\
                    & z_{u,\sigma_i} &\in& \calP_i(x_{\sigma_i}) & \text{ for all } u \in V, i \in [c]
    \end{array}
    \end{equation*}
    Here, $\calP_i(x_{\sigma_i})$ is a scaled polytope that emerges when we multiply the right-hand side of all satisfied inequalities in $\calP_i$ by $x_{\sigma_i}$. With these adjustments, the program in~\eqref{eq:LP1} becomes an LP. For $d = O(1)$, this LP has polynomial size. 
\end{proof}

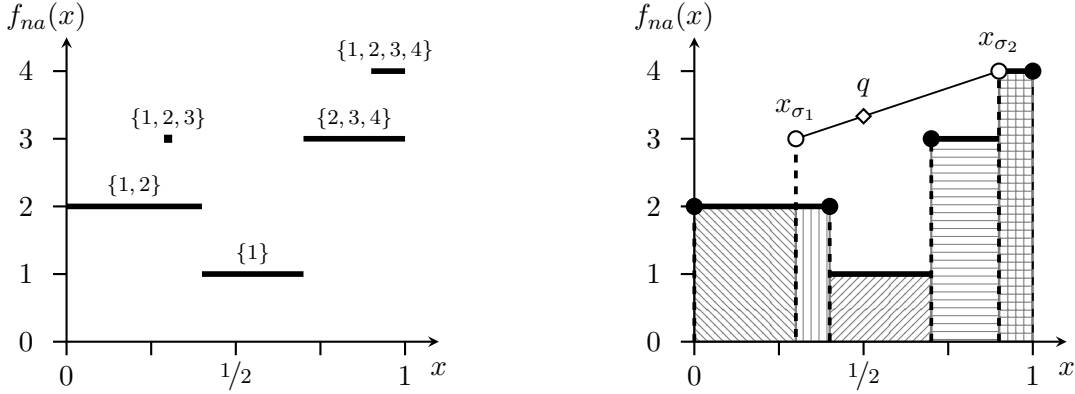
\begin{figure}[!t]
\centering
\begin{subfigure}{0.43\textwidth}
  		\begin{tikzpicture}[
				> = stealth, 
				auto,
				node distance = 3cm, 
				thick, 
				scale = .9
				]
				
				\tikzstyle{every state}=[
				draw = black,
				thick,
				fill = white,
				minimum size = 4mm
				]
                
				\path[->] (0,0.8) edge (0,5.5); 
                \path[->] (0,1) edge (5.5,1);

                \path (5/4,1) edge (5/4,.8);	
				\path (5/2,1) edge (5/2,.8);	
				\path (3*5/4,1) edge (3*5/4,.8);	
				\path (5,0.8) edge (5,1);
                
				\node at (0,0.5) {$0$};
				\node at (5,0.5) {$1$};
                \node at (5/2,0.5) {$\nicefrac{1}{2}$};
				\node at (5.5,0.6) {$x$};

                 \path (0,1) edge (-0.2,1); \node at (-0.6,1) {$0$};
                \path (0,2) edge (-0.2,2); \node at (-0.6,2) {$1$};
				\path (0,3) edge (-0.2,3); \node at (-0.6,3) {$2$};
                \path (0,4) edge (-0.2,4); \node at (-0.6,4) {$3$};
                \path (0,5) edge (-0.2,5); \node at (-0.6,5) {$4$};
                \node at (-0.3,5.8) {$f_{na}(x)$};

                \path[line width=.75mm] (0,3) edge (5*0.4,3);
                \path[line width=.75mm] (5*0.4,2) edge (5*0.7,2);
                \path[line width=.75mm] (5*0.7,4) edge (5*1,4);
                \path[line width=.75mm] (5*0.9,5) edge (5*1,5);

                \node[state,draw=black,scale=.4] at (5*0.3,4) [rectangle] { };  

                \node at (5*0.2,3.3) {\scriptsize $\lbrace 1,2 \rbrace$};  
                \node at (5*0.3,4.3) {\scriptsize $\lbrace 1,2,3 \rbrace$};   
                \node at (5*0.55,2.3) {\scriptsize $\lbrace 1 \rbrace$};  
                \node at (5*0.85,4.3) {\scriptsize $\lbrace 2,3,4 \rbrace$};  
                 \node at (5*0.94,5.3) {\scriptsize $\lbrace 1,2,3,4 \rbrace$}; 
                
			\end{tikzpicture}
            \caption{Intersection of ranges for all possible posterior beliefs $x \in [0,1]$. The sets above them  specify which agents have an opinion in one of their ranges for a given $x$.}
            \label{fig:example_ranges}
\end{subfigure}
\hspace{.5cm}
\begin{subfigure}{0.43\textwidth}%
    \centering
  		\begin{tikzpicture}[
				> = stealth, 
				auto,
				node distance = 3cm, 
				thick, 
				scale = .9
				]
				
				\tikzstyle{every state}=[
				draw = black,
				thick,
				fill = white,
				minimum size = 4mm
				]

                \draw[pattern=north west lines, pattern color=gray,draw=gray] (0,1) rectangle (5*0.3,3);

                 \draw[pattern= vertical lines, pattern color=gray,draw=gray] (5*0.3,1) rectangle (5*0.4,3);

                 \draw[pattern=north east lines, pattern color=gray,draw=gray] (5*0.4,1) rectangle (5*0.7,2);

                 \draw[pattern=horizontal lines, pattern color=gray,draw=gray] (5*0.7,1) rectangle (5*0.9,4);

                \draw[pattern=grid, pattern color=gray,draw=gray] (5*0.9,1) rectangle (5*1,5);
                
				\path[->] (0,0.8) edge (0,5.5); 
                \path[->] (0,1) edge (5.5,1);

                \path (5/4,1) edge (5/4,.8);	
				\path (5/2,1) edge (5/2,.8);	
				\path (3*5/4,1) edge (3*5/4,.8);	
				\path (5,0.8) edge (5,1);
                
				\node at (0,0.5) {$0$};
				\node at (5,0.5) {$1$};
                \node at (5/2,0.5) {$\nicefrac{1}{2}$};
				\node at (5.5,0.6) {$x$};

                 \path (0,1) edge (-0.2,1); \node at (-0.6,1) {$0$};
                \path (0,2) edge (-0.2,2); \node at (-0.6,2) {$1$};
				\path (0,3) edge (-0.2,3); \node at (-0.6,3) {$2$};
                \path (0,4) edge (-0.2,4); \node at (-0.6,4) {$3$};
                \path (0,5) edge (-0.2,5); \node at (-0.6,5) {$4$};
				\node at (-0.3,5.8) {$f_{na}(x)$};

                \path[line width=.25mm] (5*0.3,4) edge (5*0.9,5);

                \path[line width=.5mm,dashed] (5*0,1) edge (5*0,3);
                
                \path[line width=.75mm] (0,3) edge (5*0.3,3);

                 \path[line width=.5mm,dashed] (5*0.3,4) edge (5*0.3,1);
                 
                \path[line width=.75mm] (5*0.3,3) edge (5*0.4,3);

                \path[line width=.5mm,dashed] (5*0.4,3) edge (5*0.4,1);
                
                \path[line width=.75mm] (5*0.4,2) edge (5*0.7,2);

                \path[line width=.5mm,dashed] (5*0.7,1) edge (5*0.7,4);
                
                 \path[line width=.75mm] (5*0.7,4) edge (5*0.9,4);
                
                \path[line width=.5mm,dashed] (5*0.9,1) edge (5*0.9,5);
                
                \path[line width=.75mm] (5*0.9,5) edge (5*1,5);

                \path[line width=.25mm,dashed] (5*1,5) edge (5*1,1);

                \node[state,fill=black,draw=black] at (0,3) [circle] { };

            \node[state,fill=white,draw=black,label={above}:{$x_{\sigma_1}$}] at (5*0.3,4) [circle] { };

            \node[state,fill=black,draw=black] at (5*0.4,3) [circle] { };

             \node[state,fill=black,draw=black] at (5*0.7,4) [circle] { };

             \node[state,fill=white,draw=black,label={above}:{$x_{\sigma_2}$}] at (5*0.9,5) [circle] { };

             \node[state,fill=black,draw=black] at (5*1,5) [circle] { };

            \node[state,fill=white,draw=black,label={above}:{$q$}] at (5*0.5,4+1/3) [diamond] { };
               
			\end{tikzpicture}
             \caption{Resulting polytopes (dashed lines, crosshatched space in between). White circles mark the beliefs induced in the optimal scheme, the white diamond marks the prior $q$.}
           \label{fig:example_cells}
\end{subfigure}
		\caption{Illustrations for Example~\ref{ex:multid}.}
		\label{fig:constrank_example}
\end{figure}

Let us illustrate the cell decomposition approach in Theorem~\ref{thm:constantRank} with a small example for $d=2$.

\begin{example}
\label{ex:multid}
\rm Consider a setting with four agents $V=\lbrace 1,2,3,4 \rbrace$ and a (constant) number of $|\Theta| = 2$ states. By Lemma~$\ref{lem:equilib}$, we may assume that network influence is absent and all agents are fully stubborn, i.e., $\matW = \bm{0}$ and $\matD = \matI$. 

 Suppose that the full revelation matrix $\matZ$ has columns $\vecz_{\theta_1} = (0,0,0,1)^\top$ and $\vecz_{\theta_2} = (1,1,1,0)^\top$, thus $d = \text{rank}(\matZ) = 2$. We can use the decomposition $\matZ = \matB \matC$ where $\matB = \matZ$ and $\matC = \matI$. Therefore, $\vecx^{(b)}_\sigma = \vecx_\sigma/\|\vecx\|_1$ in this example. Since there are two states $\Theta = \lbrace \theta_1, \theta_2 \rbrace$, we can express any posterior belief $\vecx/\| \vecx \|_1$ over $\Theta$ in terms of the posterior probability for $\theta_2$, i.e., by a scalar $x \in [0,1]$. Thus, given signal  $\sigma \in \Sigma$, the equilibrium opinions are $z_{u,\sigma} = z_{u,\theta_1} (1-x) + z_{u,\theta_2} x = x$ for all agents $u \in V \setminus \lbrace 4 \rbrace$ and $z_{u,\sigma} = 1-x$ for $u=4$. Likewise, we write $q \in [0,1]$ for the prior distribution. 

We define the set of ranges $D_u$ of the agents to be as follows: $D_1 = \lbrace [0.0,0.7], [0.9,1.0] \rbrace$, $D_2 = \lbrace [0.0,0.4], [0.7,1.0] \rbrace $, $D_3 = \lbrace [0.3,0.3], [0.7,1.0] \rbrace$, and $D_4 = \lbrace [0.0,0.3] \rbrace$. Note that agent 3 holds a range that is equal to a single point in $[0,1]$. For consistency, we may write $z_{u,\sigma} = x$ for agent $u=4$ as well by assuming $z_{4,\theta_1} = 0$, $z_{4,\theta_2} = 1$, and $D_4 = \lbrace [0.7,1.0] \rbrace$. Evidently, this does not change the values of $x$ for which agent $4$ hits her interval. We focus on the objective $f_{na} (\vecz_\sigma) = f_{na}(x)$ of maximizing the expected number of agents with an opinion in one of their ranges. Figure~\ref{fig:example_ranges} shows the intersections of these ranges and, therefore, $f_{na}(x)$ for any posterior belief. 

Due to the choice of ranges, we can divide the $x$-space with respect to each of these intersections $[a,b]$, i.e., we introduce (in-)equalities $z_{u,\sigma} \; \Box \; \alpha$ with $\Box \in \{>, =, < \}$ and $\alpha \in \lbrace a, b \rbrace$ at both ends of each black line in Figure~\ref{fig:example_ranges}. This yields cells as depicted in Figure~\ref{fig:example_cells}: Each interval border indicated by a vertical dashed line (i.e., all points from $\{0,0.3, 0.4, 0.7, 0.9, 1\}$) is a cell on their own. The open subinterval between two consecutive dashed lines is a cell as well. Each cell $i$ is described by a (open) polytope $\calP_i$, i.e., a (open) interval in our case. Obviously, each cell is convex and has a constant value $f_{na}(\calP_i)$ for all $x \in \calP_i$. 

Next we modify the constraints describing each interval as outlined in the proof of Theorem~\ref{thm:constantRank}. We assign each boundary (vertical dashed lines) to both adjoining cells. Our optimization will then ensure that at these points, only the cell with maximum value for $f_{na}$ would be used, as indicated by the circles in Figure~\ref{fig:example_cells}. Finally, one can now compute an optimal signaling scheme $\varphi^*$ by solving our (adjusted version of) LP~\eqref{eq:LP1}. 

Notably, looking more closely at the special case of $|\Theta| = 2$, we can use a more direct, graphical approach to determine $\varphi^*$ as follows. First, we compute the expected value of $f_{na}$ at point $q$ for the no-signaling scheme as a reference. As a direct consequence of Caratheodory’s theorem (c.f. \cite{Dughmi19}), $\varphi^*$ does not have to use more than $|\Theta|$ signals. With one signal we obtain the no-signaling scheme. Hence, let us examine whether $\varphi^*$ should use two distinct signals $\Sigma = \lbrace \sigma_1,\sigma_2 \rbrace$, where $\|\vecx_\sigma \|_1 > 0$ for all $\sigma \in \Sigma$. The best scheme with two signals corresponds to posterior beliefs $x_{\sigma_1} < q < x_{\sigma_2}$ such that the line $\ell_{\sigma_1,\sigma_2}$ defined by the two points $(x_{\sigma_1}, f_{na}(x_{\sigma_1}))$ and $(x_{\sigma_2}, f_{na}( x_{\sigma_2}))$ is an upper bound for $f_{na}$ over $[x_{\sigma_1},x_{\sigma_2}]$ and has maximum value at $q$ among all such possible combinations of beliefs (cf. \cite{kamenica2011bayesian}).

From geometric considerations, we can restrict attention to the posteriors that correspond to interval borders, as part of the adjacent interval that yields larger objective value (as indicated by the circles in Figure~\ref{fig:example_cells}). We prove this property more generally in Proposition~\ref{prop:twoStates} below. In the end, we compare the maximum expected value of $\varphi$ using two signals with the expected value obtained from one signal, i.e., no-signaling. The maximum yields an optimal scheme $\varphi^*$. 

For our example, we assume $q=0.5$. There are nine possible combinations of beliefs we need to check. We find that the beliefs induced by $\varphi^*$ are $x_{\sigma_1} = 0.3$ and  $x_{\sigma_2} = 0.9$, as highlighted by the white circles in Figure~\ref{fig:example_cells}. Consequently, the optimal signaling scheme sends signal $\sigma_1$ with probability $2/3$ and signal $\sigma_2$ with probability $1/3$. This yields $\Ex \left[ f_{na}( x) \right] = 10/3 = 3.\bar{3}$ (white diamond in Figure~\ref{fig:example_cells}). For comparison, no-signaling and full revelation yield an expected value of 1 and 3, respectively.  \hfill $\blacksquare$
\end{example}

Our observations in Example~\ref{ex:multid} can be generalized to all instances with $|\Theta| = 2$ states and any range-based objective $f$. The space of posterior beliefs can be interpreted as the unit interval based on the value of $x_{\sigma,\theta_2} \in [0,1]$ (since $x_{\sigma,\theta_1} = 1-x_{\sigma,\theta_2}$). Unlike in our example, one cannot expect $z_{u,\theta_1}, z_{u,\theta_2} \in \lbrace 0, 1 \rbrace$ in general. However, we can always write $z_{u,\sigma} = x$ for the equilibrium opinion by the following transformation for each $[a,b] \in D_u$ for every agent $u \in V$: First, if $z_{u,\theta_1} > z_{u,\theta_2}$ for $u$, we can just swap their values. Correspondingly, we replace each $[a,b] \in D_u$ by the range $[1-b,1-a]$. Secondly, given that neither $z_{u,\theta_1} > b$ nor $z_{u,\theta_2} < a$, we replace any $[a,b] \in D_u$ by the interval $[a',b']$, where $a' = \max \lbrace 0,(a-z_{u,\theta_1})/(z_{u,\theta_2}-z_{u,\theta_1}) \rbrace$ and $b' = \min \lbrace 1, (b-z_{u,\theta_1})/(z_{u,\theta_2}-z_{u,\theta_1}) \rbrace$. Otherwise, $[a,b]$ can be neglected in the first place since it is never hit.
 
After assigning the boundaries, the polytopes correspond to consecutive subintervals or points in $[0,1]$. Since $f(x)$ is piecewise constant, this generates at most $2k$ \emph{breakpoints} at which the value of $f$ changes. We show in the following proposition that there is an optimal scheme $\varphi^*$ such that the induced beliefs lie exclusively on these breakpoints.

\begin{proposition}
    \label{prop:twoStates}
    For $|\Theta| = 2$ and any range-based objective $f$, there is an optimal signaling scheme that gives out signals $\sigma_1, \sigma_2$ inducing beliefs $x_{\sigma_1} \in \calP_1, x_{\sigma_2} \in \calP_2$ such that $x_{\sigma_1}$ and $x_{\sigma_2}$ lie on a breakpoint of $f$ at $\calP_1$ and $\calP_2$, respectively.
\end{proposition}
\begin{proof}
    Let $\varphi$ be a signaling scheme that induces two beliefs $x_{\sigma_1} \in \calP_1$ and $x_{\sigma_2} \in \calP_2$. The value $\Ex[f(x_\sigma)]$ of $\varphi$ is equal to
    \begin{equation}
    \label{eq:example_lineeq}
        \ell_{\sigma_1,\sigma_2}(q) = \frac{x_{\sigma_2}-q}{x_{\sigma_2}-x_{\sigma_1}}  f(\calP_1) + \frac{q-x_{\sigma_1}}{x_{\sigma_2}-x_{\sigma_1}} f(\calP_2) \, ,
    \end{equation}
    where the fractions in the first and the second term tell the probability by which signals $\sigma_1$ and $\sigma_2$ are sent, respectively. The partial derivatives of Eq.~\eqref{eq:example_lineeq} with respect to $x_{\sigma_1}$ and $x_{\sigma_2}$ read
    \begin{align}
        \frac{\partial}{\partial x_{\sigma_1}} \ell_{\sigma_1,\sigma_2} &= - \frac{\left(f(\calP_2) - f(\calP_1)\right)\left(x_{\sigma_2} - q \right)}{(x_{\sigma_2}-x_{\sigma_1})^2} \label{eq:example_deriv1} \, ,\\[1mm]
        \frac{\partial}{\partial x_{\sigma_2}} \ell_{\sigma_1,\sigma_2} &= - \frac{\left(f(\calP_2) - f(\calP_1)\right)\left(q-x_{\sigma_1} \right)}{(x_{\sigma_2}-x_{\sigma_1})^2} \label{eq:example_deriv2} \, .
    \end{align}
    First, we assume that $x_{\sigma_1} < q$ is fixed and $f(\calP_1) < f(\calP_2)$. By this assumption, we see Eq.~\eqref{eq:example_deriv2} is smaller than zero, i.e., any decrease in $x_{\sigma_2}$ along $\calP_2$ strictly increases the expected value of $\varphi$. We define $x_{b} = \min \lbrace x \mid x \in \calP_2 \rbrace$, where we may assume $x_{b} > q$ since otherwise no-signaling would yield value $f(\calP_2) \geq \Ex[f(x_\sigma)]$ in the first place. Therefore, the expected value of $\varphi$ is always improved by choosing $x_{\sigma_2} = x_b$ instead of $x_{\sigma_2} > x_b$ if $f(\calP_1) < f(\calP_2)$.

    On the other hand, if $f(\calP_1) > f(\calP_2)$, Eq.~\eqref{eq:example_deriv2} is larger than zero i.e., any increase  in $x_{\sigma_2}$ along $\calP_2$ strictly increases the expected value of $\varphi$. By a corresponding argument, $x_r < q$. Therefore, the expected value of $\varphi$ is always improved by choosing $x_{\sigma_2} = x_r$ instead of $x_{\sigma_2} < x_r$, where  $x_{r} = \max \lbrace x \mid x \in \calP_2 \rbrace$. For $f(\calP_1) = f(\calP_2)$, we can just choose $x_{\sigma_2} = x_b$ without changing the value of $\Ex[f(x_\sigma)]$.
    
    Analogously, when we assume that $x_{\sigma_1}$ can be shifted and $x_{\sigma_2} > q$ is fixed, the same observations hold through Eq.~\eqref{eq:example_deriv1}. Since $x_b$ and $x_r$ lie at the left and right border of their respective polytopes and the above findings hold for any signaling scheme, the proposition follows.
\end{proof}

Since each polytope has two endpoints, all possible $\ell_{\sigma_1,\sigma_2}$ can be tested in time $\mathcal{O}{(k^2})$. Eventually, the optimal value $\ell_{\sigma_1,\sigma_2}$ is compared to the value of no-signaling to obtain $\varphi^*$.

\subsection{Range Combinations with Non-Zero Value}
\label{sec:nonzero}

In this section, we assume that only a polynomial number of sets $T \subseteq D$ have non-zero value. This scenario contains many cases beyond the previous section, such as, e.g.:
\begin{itemize}
    \item Maximize the probability that \emph{all desired opinion intervals are hit} simultaneously, i.e., $g(D) = 1$ and $0$ else.
    \item The total number of ranges is $k \in O(\log m + \log n)$.
    \item The structure of the instance gives rise to at most polynomially many range combinations for any possible signal.
    \item The objective function is based on range combinations from a constant number of agents.
\end{itemize}

We adjust the approach above to a more direct cell decomposition. Let $\mathcal{C} =  \{ T \subseteq D \mid g(T) > 0\}$ and $\ell = |\mathcal{C}|$.

\begin{theorem}
    \label{thm:polysigscheme}
    There is an algorithm to compute an optimal signaling scheme for every range-based objective that runs in time polynomial in $n$, $m$ and $\ell$.
\end{theorem}

    \begin{equation}
    \label{eq:LP2}
    \begin{array}{rcll}
        \multicolumn{3}{l}{\text{Max}. \D \sum_{T \in \mathcal{C}} x_{\sigma_T} \cdot g(T)}\\[0.1cm]
        \text{s.t.} \D \sum_{T \in \mathcal{C} \cup \{0\}} \varphi_{\theta,\sigma_T} &=& 1 & \text{ for all }\theta \in \Theta\\
                     \varphi_{\theta,\sigma_T} & \ge & 0 & \text{ for all }\theta \in \Theta, T \in \mathcal{C} \cup \{0\} \\
                     x_{\theta,\sigma_T} &=& \varphi_{\theta,\sigma_T} \cdot q_\theta & \text{ for all } \theta \in \Theta, T \in \mathcal{C} \\
                     x_{\sigma_T} &=& \D \sum_{\theta \in \Theta} x_{\theta,\sigma_T} & \text{ for all } T \in \mathcal{C}\\
                     z_{u,\sigma_T} &=& \D \sum_{\theta \in \Theta} z_{u,\theta} x_{\theta,\sigma_T} & \text{ for all } u \in V, T \in \mathcal{C}\\
                     z_{u,\sigma_T} &\ge& a \cdot x_{\sigma_T} & \text{ for all } T \in \mathcal{C}, i_{u,[a,b]}\in T\\
                     z_{u,\sigma_T} &\le& b \cdot x_{\sigma_T} & \text{ for all } T \in \mathcal{C}, i_{u,[a,b]}\in T 
    \end{array}
    \end{equation}
    
\begin{proof}
    Reconsider the approach of Theorem~\ref{thm:constantRank} above. We now apply a cell decomposition of the space spanned by $\matZ$. Observe that for a range-based objective it is only relevant which agent is located in which interval (but not, e.g., whether an agent is above or below a certain interval).
    
    For each $T \in \mathcal{C}$ we consider necessary conditions for the equilibrium opinions, i.e., $z_{u,\sigma_T} \ge a$ and $z_{u,\sigma_T} \le b$ for each $i_{u,[a,b]} \in T$. These constraints define a convex polytope $\calP_T$. Let us first assume that each solution in $\calP_T$ has the same value $f(\calP_T) = g(T)$. Then an optimal scheme w.l.o.g.\ has at most one signal $\sigma_T$ with an opinion vector in $\calP_T$ (by convexity).

    However, the polytopes $\calP_T$ can overlap significantly. In particular, if $T \subseteq T'$ then $\calP_{T'} \subseteq \calP_T$. As such, it might be that some $\vecz \in \calP_T$ actually gives rise to a larger value than $g(T)$. Here we can again assume w.l.o.g.\ that each signal $\sigma$ with a vector $\vecz_\sigma \in \calP_{T_1} \cap \ldots \cap \calP_{T_r}$ is assigned to the one with maximal $g(T_j)$. By monotonicity of $g$, this is the one that correctly expresses the interval inclusions for all agents.

    In conclusion, LP~\eqref{eq:LP2} is similar to LP~\eqref{eq:LP1} with more explicit cells arising from the sets in $\mathcal{C}$. Note that there might be no decompositions of the prior $\vecq$ into signals $\sigma$ that \emph{all} satisfy $T_\sigma \in \mathcal{C}$. Since $g(T) = 0$ for all $T \in 2^D \setminus \mathcal{C}$, we use a single \emph{dummy signal} $\sigma_0$ for the remaining cases. The dummy signal only arises in the first two constraint sets for $\vecPhi$ to ensure that it is a feasible decomposition of $\vecq$. The size of the LP is polynomial in $n$, $m$ and $\ell$. The optimal solution represents an optimal scheme $\vecPhi$. 
\end{proof}

We elaborate on the last two example scenarios given in the beginning of this section.

\paragraph{Polynomial Signal Structures.} 
The approach can also be applied when the structure of the instance allows only a polynomial number of different range combinations to be fulfilled in any signal (even though $g$ might have arbitrarily many non-zero values). We explain the reasoning in an example scenario. Consider the case of persuading voters, i.e., each agent has a single interval $D_u = \{[0.5,1]\}$. Let agents be numbered, i.e., $V = \{1,\ldots,n\}$. 
We assume $\matZ$ yields \emph{monotone agents}: 
\begin{equation}
    \label{eq:orderedAgents}
    z_{u-1,\theta} \quad \ge \quad z_{u,\theta} \qquad \text{ for all } u \in V \setminus \{1\}, \theta \in \Theta.
\end{equation}
This captures a scenario in which voters can be ordered w.r.t.\ (pointwise non-increasing) willingness to support a yes-decision in the full revelation equilibrium of any state. Note that this scenario includes cases, in which the matrix $\matZ$ has full rank, and the number of range combinations with non-zero value are exponential in $n$ (e.g., for $g = f_{na}$).

However, observe that there are at most $n+1$ range combinations that can emerge in any signaling scheme. For any signal $\sigma$, if $z_{u,\sigma} \ge 0.5$, then~\eqref{eq:opinionDecomp} and \eqref{eq:orderedAgents} imply $z_{u-1,\sigma} \ge 0.5$ as well. Consequently, we either have no agent with an opinion in $[0.5,1]$ at all, or there is a unique largest agent $u_\sigma$ such that exactly the agents $\{1,\ldots,u_\sigma\}$ have an equilibrium opinion in $[0.5,1]$. We can limit attention to a set $\mathcal{C}$ containing these $n$ range combinations and apply Theorem~\ref{thm:polysigscheme}. 

The restriction on signal structure is present in many more general scenarios. For example, for monotone agents with monotone ranges $D_u = \{[a_u,1]\}$ with $a_{u-1} \le a_u$ for all $u \in V \setminus \{1\}$, precisely the same argument applies. In a different direction, assume that $\matZ$ yields \emph{bitonic agents}, where there is $k \in \{2,\ldots,n\}$ such that
\begin{align*}
    z_{u-1,\theta} \quad &\ge \quad z_{u,\theta} \qquad \text{ for all } u \in \{2,\ldots,k\}, \theta \in \Theta, \text{ and }\\
    z_{u-1,\theta} \quad &\le \quad z_{u,\theta} \qquad \text{ for all } u \in \{k+1,\ldots,n\}, \theta \in \Theta.
\end{align*}
Consider any signal $\sigma$. By the same arguments as above, if $z_{k,\sigma} \ge 0.5$, then $z_{v,\sigma} \ge 0.5$ for all $v \in V$. Otherwise, there exists a unique largest agent among $\{1,\ldots,k-1\}$ such that $z_{u,\sigma} \ge 0.5$ (or none), as well as a unique smallest agent among $\{k+1,\ldots,n\}$ (or none). Hence, in total there are at most $1+k(n-k+1) \in O(n^2)$ range combinations that can emerge in any signaling scheme. 

We point out that the result of Theorem~\ref{thm:constantRank} also falls into this regime. Due to constant rank of $\calS$, we can restrict attention to the $O(k^d)$ many range combinations that result from all the cells $\calP_i$. 

\paragraph{Ranges from a Constant Number of Agents.}
Suppose the objective function $g$ is defined only based on the combinations of ranges for $d$ agents. Clearly, we can disregard the ranges for all remaining agents. For optimizing $g$ we are only interested in equilibrium opinions of the $d$ agents. As such, we can limit attention to computing the set of posteriors, i.e., linear combinations of the rows of $\matZ$ corresponding to the $d$ agents. Note that $\matS$ might even have full rank. However, by restricting $\matZ$ to $d$ rows, we obtain an instance with rank $d$ of $\matZ$. The proof of Theorem~\ref{thm:constantRank} can be applied to show that we only need to consider $O(k^d)$ many range combinations (and, by our observation above, this allows to apply Theorem~\ref{thm:polysigscheme}).

\subsection{General Case}
\label{sec:general}

In this section, we first observe that we can efficiently compute a scheme that represents a (multiplicative) $n$-approximation. This applies for all range-based objectives that are subadditive, i.e., for which $g(T \cup T') \le g(T) + g(T')$. We show that this is tight -- finding an optimal signaling scheme can be \classNP-hard, even for the additive objective $f_{na}$. It is \classNP-hard to approximate the optimal signaling scheme within $n^{1-\varepsilon}$, for any constant $\varepsilon > 0$. 

\begin{theorem}
    For every subadditive range-based objective, a signaling scheme that represents an $n$-approximation can be computed in polynomial time.
\end{theorem}

\begin{proof}
    Consider each agent $u \in V$ individually, and compute a signaling scheme $\vecPhi^{(u)}$ that maximizes $\Ex[g(T_{\sigma} \cap D_u)]$, i.e., the expected value of the ranges of $D_u$ that contain $z_{u,\sigma}$. To do this, we only need to consider the range combinations resulting from $D_u$ w.r.t.\ $z_{u,\sigma}$, which are at most $2 |D_u|$. By our arguments at the end of the last section, we can apply Theorem~\ref{thm:constantRank} to compute each $\vecPhi^{(u)}$ in polynomial time. Linearity of expectation and subadditivity of $g$ imply that the best of these $n$ schemes is an $n$-approximation.
\end{proof}

\begin{theorem}
\label{thm:np-hardness}
    For maximizing $f_{na}$, it is \classNP-hard to compute a signaling scheme that represents a $n^{1-\varepsilon}$-approximation, for any constant $\varepsilon > 0$.
\end{theorem}

\begin{proof}
    We reduce from the \textsc{Independent Set} problem. Given an undirected graph $H=(V,E_H)$, where $|V| = n$, 
    we construct an instance of FJ signaling with network $G=(V,\emptyset)$, state set $\Theta = \{\theta_u \mid u \in V\}$, prior $q_{\theta_u} = 1/n$ for all $\theta_u \in \Theta$, and preconception matrix $\matS$ as follows. The entries $\vecs_{u,\theta_u}=\frac{3}{4}$ for all $u\in V$. For each edge $\{u,v\}\in E_H$, $\vecs_{u,\theta_v}=\vecs_{v,\theta_u}=0$. All remaining entries of $\matS$ are $\frac{1}{2}-\frac{1}{8n}$. The desired range is $D_u = \{[0.5,1]\}$ for all $u \in V$. Note that $\matS = \matZ$ since $E = \emptyset$. Our construction is similar to~\cite[Observations 8.1 and 8.2]{cheng2015mixture}.
    
    Consider any independent set $\mathcal{I}$ of $H$. Let $k=|\mathcal{I}|$. There exists a signal $\sigma$ such that the corresponding entries of $\vecz_\sigma = \matZ \vecx_\sigma / \|\vecx_\sigma\|_1$ are at least $\frac{1}{2}+\frac{1}{8n}$, while all other entries are less than $\frac{1}{2}$. Formally, suppose $x_{\theta_u,\sigma} = \frac{1}{n}$ (i.e., $\varphi_{\theta_u,\sigma}=1$) for each $u\in \mathcal{I}$, and $x_{\theta_u, \sigma}=0$ (i.e., $\varphi_{\theta_u,\sigma}=0$) otherwise. If $u\in \mathcal{I}$, $\vecz_{u,\sigma} = \frac{1}{k}(\frac{3}{4} + (k-1)(\frac{1}{2}-\frac{1}{8n})\ge \frac{1}{2}+\frac{1}{8n}$, and $\vecz_{u,\sigma}\le \frac{1}{2}-\frac{1}{8n}$ otherwise.
    
    Conversely, for any $\sigma$, the entries of $\vecz_\sigma$ in the range $[0.5,1]$ must correspond to an independent set of $H$. For an edge $\{ u,v\}\in E_H$, w.l.o.g.\ assume that $x_{\theta_u,\sigma} \ge x_{\theta_v,\sigma}$. If $x_{\theta_u,\sigma}=0$, then $\vecz_{u,\sigma}\le \frac{1}{2}-\frac{1}{8n} < \frac{1}{2}$, else $\vecz_{v,\sigma}\le \frac{3\cdot x_{\theta_v,\sigma}}{4\cdot(x_{\theta_v,\sigma}+x_{\theta_u,\sigma})} < \frac{1}{2}$. $\vecz_{u,\sigma}$ and $\vecz_{v,\sigma}$ cannot \emph{both} be in $[0.5,1]$.
    
    It is \classNP-hard to approximate \textsc{Independent Set} to within a factor of $n^{1-\delta}$, for any constant $\delta > 0$. The hardness arises in distinguishing graphs with independent sets of size at least $n^{\delta/2}$ from ones where all independent sets have size at most $n^{(1-\delta)\delta/2}$~\cite[Theorem 1.1]{Zuckerman07}.
    
    If all independent sets have a size smaller than $n^{\delta/2}$, the expected amount of agents in their range is at most $n^{\delta/2}$, because the maximum amount of agents in their range is $n^{\delta/2}$, and the maximum probability of all signals is $1$. If at least one independent set $\mathcal{I}$ has size $n^{\delta/2} \cdot n^{1-\delta}$, a signal $\sigma$ exists, where $n^{\delta/2} \cdot n^{1-\delta}$ agents are in the desired range $[0.5,1]$. This signal has a probability of at least $\|\vecx_\sigma \|_1 \ge n^{\delta/2} n^{1-\delta}/n$ because for each $u \in \mathcal{I}$, in state $\theta_u$ the signal $\sigma$ is sent with $\varphi_{\theta_u,\sigma}=1$. The expected number of agents in their ranges is then at least $(n^{\delta/2} n^{1-\delta})^2/n$. This results in a lower bound of $\frac{(n^{\delta/2}n^{1-\delta})^2/n}{n^{\delta/2}} = n^{1-\frac{3}{2}\delta} = n^{1-\varepsilon}$, for any constant $\varepsilon > 0$. 
\end{proof}

\section{Extensions}
\label{sec:extensions}
In this section, we discuss the extension of our results in Sections~\ref{sec:model} and \ref{sec:simple} to variants of the classic FJ opinion formation model we have considered so far.

First, our results also apply to a multidimensional extension of the FJ model discussed by Parsegov et al.~\cite{parsegov2016novel}. They assume each agent $u$ has vector-valued preconceptions and public opinions $\vecs_u, \vecz^{t}_u \in \RR^q$. The elements of vector $\vecz^{t}_u = (z^{t}_{u,1}, z^{t}_{u,2},..., z^{t}_{u,q})$ correspond to the opinions of agent $u$ on $q > 0$ different \emph{issues}. The agents apply the FJ opinion formation process on these $q$ issues simultaneously, where the update on issue $i$ may be unrelated or interdependent of other issues $j \neq i$. It is easy to verify that this multidimensional setting can be reduced to the standard case by using an interpretation with \emph{auxiliary} agents, one for each pair $(u,h)$ of agent $u \in V$ and issue $h \in \{1,...,q\}$. All our results extend accordingly also to this more general scenario.

In our discussion above, we assumed that instances of FJ signaling satisfy standard convergence criteria. These criteria require a suitable set of agents to have $\lambda_u < 1$ and are usually violated if $\matLambda = \matI$. The latter scenario represents the classic French-DeGroot model. Here, the update~\eqref{eq:FJ_update} of the public opinion of each agent $u \in V$ in time step $t > 0$ is solely governed by the current public opinions of its neighbors. 

Public opinions become $\vecz^{(t+1)} = \matI\matA \vecz^{(t)} = \matW \vecz^{(t)} = (\matW)^{t+1}\vecs$. The analysis of this process has close connections to Markov chains, see, e.g., \cite[Chapter 3]{ProskunikovT2017}. Interestingly, if the network is strongly connected and aperiodic, all agents converge to a single consensus opinion. Formally, consensus is reached if and only if there exists a row-vector $\vecpi=(\pi_u)_{u \in V}$ s.t.\ $\matW^{\infty} = \lim_{t\rightarrow \infty} \matW^t$ exists and every row of $\matW^{\infty}$ is $\vecpi$. In this case $\matW^\infty = (\matI - \matW)^{-1}$, and $\pi_u \in [0,1]$ for all $u \in V$ and $\sum_{u} \pi_u = 1$. The consensus opinion adopted by all agents is $z = \vecpi \vecs$.

Similar to FJ signaling, we consider \emph{FDG signaling}, i.e., Bayesian persuasion in instances of the French-DeGroot model that converge to consensus. Notably, preconceptions in $\matS$ are now used only \emph{initially} as starting opinions for the convergence process (but they still govern the emerging consensus opinion). Public opinions in equilibrium reduce to a single number, the consensus opinion $z$. Let $\vecz = (z_{\theta})_{\theta \in \Theta}$ be the consensus opinions that evolve for full revelation signaling in the individual states, i.e., $\vecz = \vecpi \matS$. More generally, due to consensus, the full revelation matrix $\matZ = (\matI-\matW)^{-1}\matS$ has every row identical to $\vecz$ and, thus, rank 1. At this point, it is clear that the arguments from Theorems~\ref{thm:convex} and~\ref{thm:constantRank} can be applied to this scenario.

\begin{corollary}
    If $\calS$ wants to minimize a convex objective function $f(\vecz_\sigma)$ in FDG signaling, then no-signaling is an optimal scheme. If $\calS$ wants to maximize such a function, then full revelation is an optimal scheme.
\end{corollary}

\begin{corollary}
    For FDG signaling, there is a polynomial-time algorithm to compute an optimal signaling scheme for every range-based objective.
\end{corollary}

\bibliographystyle{abbrv}
\bibliography{references}

\end{document}